\documentclass[conference]{IEEEtran}
\IEEEoverridecommandlockouts
\usepackage{cite}
\usepackage{amsmath,amssymb,amsfonts}
\usepackage{algorithmic}
\usepackage{graphicx}
\usepackage{textcomp}
\usepackage{xcolor}
\usepackage{cite}
\usepackage{amsmath,amssymb,amsfonts}
\usepackage{algorithmic}
\usepackage{graphicx}
\usepackage{textcomp}
\usepackage{xcolor}
\usepackage{amsmath, amssymb, amsthm}

\newtheorem{definition}{Definition}
\newtheorem{theorem}{Theorem}
\newtheorem{lemma}{Lemma}

\newtheorem{corollary}{Corollary}

\def\BibTeX{{\rm B\kern-.05em{\sc i\kern-.025em b}\kern-.08em
    T\kern-.1667em\lower.7ex\hbox{E}\kern-.125emX}}
\begin{document}

\title{Generalized Score Matching: Bridging $f$-Divergence and Statistical Estimation Under Correlated Noise\\

}

\author{\IEEEauthorblockN{Yirong Shen}
\IEEEauthorblockA{\textit{Electrical \& Electronic Engineering} \\
\textit{Imperial College London}\\
London SW7 2AZ, UK \\
yirong.shen22@imperial.ac.uk}
\and
\IEEEauthorblockN{Lu Gan}
\IEEEauthorblockA{\textit{Electronic \& Electrical Engineering} \\
\textit{Brunel University of London}\\
Uxbridge UB8 3PH, UK \\
lu.gan@brunel.ac.uk}
\and
\IEEEauthorblockN{Cong Ling}
\IEEEauthorblockA{\textit{Electrical \& Electronic Engineering} \\
\textit{Imperial College London}\\
London SW7 2AZ, UK \\
c.ling@imperial.ac.uk}

}

\maketitle

\begin{abstract}
Relative Fisher information, also known as score matching, is a recently introduced learning method for parameter estimation. Fundamental relations between relative entropy and score matching have been established in the literature for scalar and isotropic Gaussian channels. This paper demonstrates that such relations hold for a much larger class of observation models. We introduce the vector channel where the perturbation is non-isotropic Gaussian noise. For such channels, we derive new representations that connect the $f$-divergence between two distributions to the estimation loss induced by mismatch at the decoder. This approach not only unifies but also greatly extends existing results from both the isotropic Gaussian and classical relative entropy frameworks. Building on this generalization, we extend De Bruijn's identity to mismatched non-isotropic Gaussian models and demonstrate that the connections to generative models naturally follow as a consequence application of this new result.
\end{abstract}

\section{Introduction}
Parameter estimation plays a central role in machine learning and statistics. In high-dimensional models, \textit{maximum likelihood estimation} (MLE) is often computationally prohibitive due to the complexity of partition functions. As a practical alternative, \textit{relative Fisher information} \cite{verdu2010mismatched}, also known as \textit{score matching} \cite{hyvarinen2005estimation} emerges as a powerful technique, particularly effective for learning unnormalized statistical models. 


Recent advances have proposed practical solutions to address computational challenges in score matching. Methods such as approximate backpropagation \cite{kingma2010regularized} and curvature propagation \cite{10.5555/3042573.3042698} aim to reduce the cost associated with computing the trace of the Hessian, while alternative formulations interpret score matching as a denoising problem \cite{6795935}. Moreover, score-based generative models \cite{DBLP:conf/iclr/0011SKKEP21,song2019generative} have popularized the use of score matching as a training objective, demonstrating its robustness and effectiveness. These methods highlight that score matching is naturally resilient to noisy data, as theoretically connected to MLE under isotropic Gaussian noise \cite{10.5555/1795114.1795156}.


However, the assumption of isotropic noise, a cornerstone of traditional score matching, severely restricts its applicability to more complex noise environments often encountered in modern generative models. Recent developments have drawn increasing attention to the critical role of correlated and structured noise. Applications such as image editing \cite{yu2024constructing}, non-isotropic noise modeling \cite{huang2024blue}, and blurring diffusion models \cite{hoogeboom2023blurring} exemplify the growing demand to extend score matching techniques beyond the isotropic setting. This highlights the need for a generalized framework capable of bridging score matching with broader and more realistic noise structures.

Furthermore, score matching’s intuitive connection to log-likelihood (or relative entropy) reinforces its role as a fundamental objective for probabilistic generative models \cite{DBLP:journals/corr/TheisOB15}. Nevertheless, alternative objectives explored in implicit generative modeling \cite{farnia2016minimax} suggest that relaxing the strict reliance on log-likelihood could unlock improved performance and new possibilities in generative modeling.

\textbf{Our contributions}: Motivated by these challenges, we propose \textit{f score matching} a novel variant that generalizes traditional score matching along two important directions:
\begin{itemize}
    \item It generalizes the learning objective to a broad family of $f$-divergence, offering a flexible and principled alternative to relative entropy.
    \item It extends score matching to correlated noise structures, thereby overcoming the limitations imposed by isotropic noise assumptions.
\end{itemize}
In addition to its theoretical formulation, $f$-score matching retains favorable properties for numerical optimization. In algorithmic design, incorporating the Hessian matrix alongside gradients (as in Newton’s method \cite{5165186}), can significantly accelerate convergence. Our framework naturally accommodates this by aligning excess score functions along noise perturbations.

Furthermore, the generalization provided by \textit{f score matching} enables robust and tractable analysis in models such as Gaussian mixture learning \cite{jia2023entropic} and non-isotropic noise perturbation diffusion \cite{yu2024constructing}. It unifies various existing learning objectives under a common theoretical foundation and provides a scalable approach for handling complex noise in high-dimensional generative tasks. The resulting framework finds broad applications, including adversarial training \cite{goodfellow2014generative, farnia2018convex, nowozin2016f}, optimal covariance matching, robust optimization, and density estimation for generative models with correlated noise.

Finally, building on \cite{1564431}, we extend De Bruijn’s identity to mismatched vector Gaussian channels and analyze the interplay between Fisher information, relative entropy, and mismatched estimation as a concrete instance within the $f$-divergence family. We also discuss potential applications of these results, underscoring the broader impact and versatility of the developed theoretical framework.

The remainder of the paper is organized as follows. Section \ref{problem setting} describes the problem setting and key terminology. Section \ref{main results} presents our main results on generalized score matching, including cases with correlated Gaussian noise and classical relative entropy. Section \ref{implications} discusses the implications for learning models and concludes the paper. Proofs of technical lemmas are deferred to Section \ref{appendix}.


\section{Background}\label{problem setting}

In this section, we establish the theoretical foundation for the proposed \textit{f score matching} framework by revisiting key concepts in information theory and statistical estimation. These include relative Fisher information, which quantifies the gradient mismatch between distributions, $f$-divergence, a versatile measure for distributional differences, and vector Gaussian channels, which model the impact of noise in generative tasks.

\subsection{Relative Fisher Information Matrix}

A key result from \cite{5508632} is the idea of relative Fisher information, which had previously received limited attention. Besides the Fisher information matrix \cite{Fisher_1925}, another quantity that is closely related to the relative entropy \cite{cover1999elements} is the relative Fisher information matrix. Formally, the relative Fisher information matrix between two probability density functions (pdfs) \(p(\mathbf{x})\) and \(q(\mathbf{x})\)\footnote{In this paper, random objects are denoted by uppercase letters, and their realizations by lowercase letters. The expectation \(E(\cdot)\) is taken over the joint distribution of the random variables inside the brackets.} is defined as:
\begin{definition}
    ({Relative Fisher Information Matrix})
    Let \(p(\mathbf{x})\) and \(q(\mathbf{x})\) be two densities of a random vector \(\mathbf{x}\) (with finite second-order moments). The relative Fisher information matrix with respect to a translation parameter is:
    \begin{equation}\label{relative fisher information Eq}
        \mathcal{I}\bigl(p\!\parallel\!q\bigr) 
        = E_{p(\mathbf{x})}\Bigl[\|\nabla_{\mathbf{x}}\log p(\mathbf{x}) 
        \;-\;\nabla_{\mathbf{x}}\log q(\mathbf{x})\|_{\mathrm{OP}}^{2}\Bigr],
    \end{equation}
    where the notation \(\|\mathbf{u}\|^2_{\mathrm{OP}}\) denotes the outer product of a vector \(\mathbf{u} \in \mathbb{R}^d\) with itself, that is,\[\|\mathbf{u}\|^2_{\mathrm{OP}} := \mathbf{u} \mathbf{u}^\top \in \mathbb{R}^{d \times d}.\]
\end{definition}
It is related to score matching \cite{hyvarinen2005estimation} because \(\nabla_{\mathbf{x}} \log p(\mathbf{x})\) is known as the score function in statistics.
The term “score” here refers to the gradient of the log density with respect to the data vector (rather than the derivative of the log likelihood with respect to parameters). 
This measure avoids computing the partition function by using a different metric on the density functions, which makes it especially suitable for estimating the unnormalized models, where the partition function is denoted as $Z_{\boldsymbol{\theta}}$, which is assumed to be existent but intractable.

\subsection{Learning Metric}

In statistical models, we have \(d\)-dimensional data \(\mathbf{X}\) with density \(p(\mathbf{x})\). The goal is to fit a parametric probabilistic model \(q(\mathbf{x};{\boldsymbol{\theta}})\), where \(\boldsymbol{\theta}\) is the model parameter, to match \(p(\mathbf{x})\) as closely as possible. Specifically, one selects a divergence between two pdfs, \(p(\mathbf{x})\) and \(q(\mathbf{x};{\boldsymbol{\theta}})\). Learning then involves finding the parameter \(\boldsymbol{\theta}\) from some parameter space \(\boldsymbol{\Theta}\) that minimizes this divergence:
\begin{equation}\label{distance}
    \min \ell \bigl(p(\mathbf{x}),\, q(\mathbf{x};{\boldsymbol{\theta}})\bigr),
\end{equation}
where \(\ell(\cdot,\cdot)\) is the chosen divergence. Ideally, the divergence metric should be non-negative and equal to zero only when the two densities are identical almost everywhere.

Building on this general view, among the information measures, \( f \)-divergence stands out as a general and versatile metric, as it is defined for distributions that are discrete, continuous, or neither. Moreover, other information measures, such as relative entropy (Kullback–Leibler divergence), Jensen–Shannon divergence and squared Hellinger distance, can be easily expressed in terms of \( f \)-divergence  \cite{csiszar1967information}.
\begin{definition}\label{f divergence def}
    (\(f\)-divergence) Let \(f: (0, \infty) \to \mathbb{R}\) be a strictly convex, continuous function with \(f(1) = 0\). Let \(P\) and \(Q\) be probability distributions on a measurable space \((\mathcal{X}, \mathcal{F})\). If \(P \ll Q\), the \(f\)-divergence is:
\begin{equation}\label{f-divergence eq}
    D_{f}(P\|Q) 
    = \int f\!\Bigl(\frac{dP}{dQ}\Bigr)\,dQ,
\end{equation}
where \(\tfrac{dP}{dQ}\) is the Radon–Nikodym derivative, and \(f(0) \triangleq f(0+)\).
\end{definition}

\subsection{Probabilistic Diffusion Models}\label{mismatched channel section}

Recently, diffusion generative models \cite{ho2020denoising} are often modeled by the vector Gaussian channel \cite{1412024}. In this work, we generalize Gaussian diffusion models by allowing a non-isotropic Gaussian noise distribution, whose covariance matrix \(\mathbf{\Sigma}\) is positive semi-definite, rather than the identity matrix:
\begin{equation}\label{vector channel}
    \mathbf{Y} 
    = \sqrt{\tilde{\alpha}_t}\,\mathbf{H}\,\mathbf{X} 
    + \sqrt{1 - \tilde{\alpha}_t}\,\mathbf{N},
\end{equation}
where \(\mathbf{H} \in \mathbb{R}^{m\times d}\) is a channel matrix, $\tilde{\alpha}_t(t) \in \mathbb{R}$ is the coefficient controlling by a time parameter $t \in [0,T]$ and \(\mathbf{N} \sim \mathcal{N}(\mathbf{n};\, 0,\mathbf{\Sigma})\) is additive Gaussian noise with an \(m\times m\) covariance matrix \(\mathbf{\Sigma}\). The input \(\mathbf{X}\) follows a distribution \(P_{\mathbf{X}}\), and the output \(\mathbf{Y}\) follows \(P_{\mathbf{Y}}\), where \(\mathbf{X}\) and \(\mathbf{Y}\) are column vectors of appropriate dimensions. 

\section{Main Results}\label{main results}

\subsection{Generalized Score Matching}

There is a striking similarity between the relative Fisher information and the $f$-divergence as in Eq.~\eqref{f-divergence eq}. If we rewrite the Fisher divergence, Eq.~\eqref{relative fisher information Eq}, as a generalized form:
\begin{align}
    \mathcal{I}_{f}(P\|Q) &= E_{Q} \Bigl[f''\Bigl(\frac{dP}{dQ}\Bigr)\,\bigl\|\nabla\frac{dP}{dQ}\bigr\|_{\text{OP}}^{2}\Bigr] \\
               &= E_{P}\Bigl[\frac{dP}{dQ}\,f''\Bigl(\frac{dP}{dQ}\Bigr)\,\bigl\|\nabla\log\frac{dP}{dQ}\bigr\|_{\text{OP}}^{2}\Bigr],
\end{align}
the difference between them lies in the fact that, instead of utilizing the likelihood ratio directly, the Fisher divergence family calculates the mismatched estimation errors of the gradient of the likelihood ratio. This suggests a potential general connection between them, and consequently, between score matching and information measures. This relationship is indeed significant and is encapsulated in the following theorem.

\begin{theorem}\label{f divergence}
    Consider the signal model in \eqref{vector channel}, where \(\mathbf{X}\) (with finite second-order moments) is arbitrary, and \(\mathbf{N}\) is Gaussian with covariance matrix \(\mathbf{\Sigma}\), independent of \(\mathbf{X}\). Denote by \(p(\mathbf{y})\) and \(q(\mathbf{y};{\boldsymbol{\theta}})\) the densities of \(\mathbf{Y}\) when \(\mathbf{X}\) has distribution \(p(\mathbf{x})\) or \(q(\mathbf{x};{\boldsymbol{\theta}})\), respectively. Assume \(p(\mathbf{y})\) and \(q(\mathbf{y};{\boldsymbol{\theta}})\) are smooth and decay quickly, so that their logarithms have growth at most polynomial at infinity. 
    Let \(f(\cdot)\) have a second derivative \(f''(\cdot)\), and assume \(D_{f}(p(\mathbf{y})\|q(\mathbf{y};{\boldsymbol{\theta}}))\) is finite. Then:
\begin{equation}\label{f divergence equation}
    \nabla_{\mathbf{\Sigma}} \,D_{f}\bigl(p(\mathbf{y})\,\|\;q(\mathbf{y};{\boldsymbol{\theta}})\bigr) 
    = -\frac{1}{2}\, \mathcal{I}_{f}\bigl(p(\mathbf{y})\,\|\;q(\mathbf{y};{\boldsymbol{\theta}})\bigr).
\end{equation}
\end{theorem}

To prove Theorem \ref{f divergence}, we first require two lemmas, the proofs of which are provided in the Appendix. 
\begin{lemma}\label{heat equation}
    (Matrix Heat Equation)
    Let an input signal \(\mathbf{X}\) with source density \(p(\mathbf{x})\) pass through the channel in \eqref{vector channel}, producing an output \(\mathbf{Y}\) with density \(p(\mathbf{y})\). Then for every \(\mathbf{Y} \in \mathbb{R}^{m}\), we have
    \begin{equation}\label{matrix heat equation}
        \nabla_{\mathbf{\Sigma}} p(\mathbf{y}) 
        = \frac{1}{2}\,\mathbb{H}_{\mathbf{y}}\!\bigl(p(\mathbf{y})\bigr),
    \end{equation}
    where \(\nabla_{\mathbf{\Sigma}}\) denotes the gradient with respect to \(\mathbf{\Sigma}\), and \(\mathbb{H}_{\mathbf{y}} := \Delta_{\mathbf{y}} \) is the Hessian operator \footnote{The Hessian operator \(\mathbb{H}\bigl(f(\cdot)\bigr)\), sometimes denoted by \(\Delta\) or \(\nabla^2\) (with a slight abuse of notation in this paper), is the transpose of the Jacobian matrix of the gradient of the function \(f(\cdot)\).} with respect to \(\mathbf{y}\). 
\end{lemma}
\begin{lemma}\label{integration vanish}
    Let $f(\cdot)$ be same defiend as in Definition~\ref{f divergence def}, for any function $q(\mathbf{y};{\boldsymbol{\theta}})f(\tfrac{p(\mathbf{y})}{q(\mathbf{y};{\boldsymbol{\theta}})})$ whose gradient $\nabla_{\mathbf{y}}$ and Hessian operator $\Delta_{\mathbf{y}}$ are well defined, the following integral vanishes:
    \begin{align}\label{integration vanish equation}
        \int_{\mathbb{R}^m} \nabla_{\mathbf{y}}\,\cdot\,\bigl[\nabla_{\mathbf{y}}\bigl(q(\mathbf{y};{\boldsymbol{\theta}})\,f(\frac{p(\mathbf{y})}{q(\mathbf{y};{\boldsymbol{\theta}}
        )})\bigr)\bigr] \,d\mathbf{y}
\;=\; 0.
    \end{align} 
\end{lemma}

\begin{proof}
    (Proof of Theorem \ref{f divergence})  
    For brevity, we omit references to the variable \(\mathbf{y}\) and the parameter \(\boldsymbol{\theta}\) in the integrals and density functions, whenever this does not cause ambiguity.

    Let \(f(\cdot)\) be a single-variable function, with \(f'\) and \(f''\) denoting its first and second derivatives, respectively. We seek \(\nabla_{\mathbf{\Sigma}}\,D_f(p \| q)\). This is written as: 
    \begin{align}
        \nabla_{\mathbf{\Sigma}} D_f(p \| q) &= \int_{\mathbb{R}^m} \left[f\left(\frac{p}{q}\right) - \left(\frac{p}{q}\right) f'\left(\frac{p}{q}\right)\right] \nabla_{\mathbf{\Sigma}} q \, d\mathbf{y} \notag \\
        &\quad + \int_{\mathbb{R}^m} f'\left(\frac{p}{q}\right) \nabla_{\mathbf{\Sigma}} p \, d\mathbf{y}.
    \end{align}

    From channel \eqref{vector channel}, the density \(q(\mathbf{y})\) can be expressed as \(q(\mathbf{y}) = \int q(\mathbf{x}) p(\mathbf{y}|\mathbf{x})\,d\mathbf{x}\). According to Lemma \ref{heat equation}, the gradient with respect to \(\mathbf{\Sigma}\) of \(q(\mathbf{y})\) is given by \(\nabla_{\mathbf{\Sigma}} q(\mathbf{y}) = \frac{1}{2} \mathbb{H}_q(\mathbf{y})\). Employing the matrix form of the heat equation and simplifying by omitting references to the variable \(\mathbf{y}\) in the operators, we derive the following:
    \begin{align}\label{Hessian matrix of f}
        \nabla_{\mathbf{\Sigma}} D_f(p \| q) &= \frac{1}{2} \int_{\mathbb{R}^m} \left[f\left(\frac{p}{q}\right) - \left(\frac{p}{q}\right) f'\left(\frac{p}{q}\right)\right] \Delta q \, d\mathbf{y} \notag \\
        &\quad + \frac{1}{2} \int_{\mathbb{R}^m} f'\left(\frac{p}{q}\right) \Delta p \, d\mathbf{y}.
    \end{align}
    Utilizing integration by parts, \textit{i.e.,} \( g'h= (gh)' - gh' \), as adopted in \cite{chen2013mismatched,5165186,1564431} and further elaborated in the simplification techniques of \cite{guo2009relative}, the integrand can be equivalently transformed as follows:
    \begin{equation}\label{last term}
        \Delta \left(q f\left(\frac{p}{q}\right)\right) - \nabla p \cdot \Delta f\left(\frac{p}{q}\right) + \left(\frac{p}{q}\right) \nabla q \cdot \Delta f\left(\frac{p}{q}\right).
    \end{equation}
    
    Applying Lemma \ref{integration vanish}, the first term in Eq. \eqref{last term} vanishes upon integration:
    \begin{equation}
        \frac{1}{2}\int_{\mathbb{R}^m} \nabla \cdot \left[ \nabla \left( q f\left(\frac{p}{q}\right)\right) \right] \, d\mathbf{y} = 0,
    \end{equation}
    leaving the following two terms:
    \begin{align}
        &- \nabla p \cdot \Delta f\left(\frac{p}{q}\right) + \left(\frac{p}{q}\right) \nabla q \cdot \Delta f\left(\frac{p}{q}\right) \notag \\
        &= -q \nabla \left(\frac{p}{q}\right) \cdot \Delta f\left(\frac{p}{q}\right) =  -q\nabla \left(\frac{p}{q}\right) \bigl(\nabla \bigl(f'\left(\frac{p}{q}\right)\bigr)\bigl)^T \notag\\
        &= -q f''(\frac{p}{q})\left(\nabla(\frac{p}{q})\nabla^{T}(\frac{p}{q})\right) \notag\\
        &= -q \left\|\nabla \left(\frac{p}{q}\right)\right\|_{\text{OP}}^{2} f''\left(\frac{p}{q}\right). \notag
    \end{align}
    
    Collecting the results above, we obtain:
    \begin{equation}
        \nabla_{\mathbf{\Sigma}} D_f\left( p(\mathbf{y}) \| q_{\boldsymbol{\theta}}(\mathbf{y}) \right) = -\frac{1}{2} \mathcal{I}_f \left( p(\mathbf{y}) \| q_{\boldsymbol{\theta}}(\mathbf{y}) \right),
    \end{equation}
    which completes the proof of Theorem \ref{f divergence}.
\end{proof}
Theorem \ref{f divergence} establishes a formal connection between $f$-score matching and $f$-divergence, revealing intriguing aspects of their interplay. This result generalizes the relation between mismatched estimation and information measures discussed in \cite{10.5555/1795114.1795156} to vector Gaussian channels with correlated noise, in a manner analogous to the generalization from \cite{verdu2010mismatched} to \cite{chen2013mismatched}. Specifically, we show that, in vector Gaussian channels, the gradient of the $f$-divergence between the channel output distributions with respect to the noise covariance matrix is directly linked to the excess relative Fisher information matrix induced by a mismatched input distribution. A similar relation for scalar Gaussian channels was previously established in \cite{guo2009relative}. Since Fisher divergence plays a central role in score matching methods \cite{huang2021variational,song2019generative}, the connection uncovered here may provide new insights for density estimation and optimization tasks in (implicit) generative modeling. Potential applications are expected to emerge in this direction, particularly given the relevance of divergences such as the Jensen-Shannon divergence and the squared Hellinger distance in various optimization problems in machine learning \cite{farnia2018convex,10446269,shi2024on}.

\subsection{Classical Relative Entropy}

It is widely acknowledged that maximizing a model's log-likelihood corresponds to minimizing the Kullback-Leibler (KL) divergence between the data distribution and the model distribution. Nonetheless, reducing score matching loss does not necessarily affect the likelihood in a comparable manner. Score-based diffusion models exploit the alignment between KL divergence and score matching under isotropic Gaussian noise to approximate maximum likelihood \cite{song2021maximum}, demonstrating significant utility in applications such as compression \cite{theis2022lossy}, semi-supervised learning \cite{lasserre2006principled, dai2017good}, and adversarial purification \cite{song2017pixeldefend}. Whereas maximum likelihood directly targets the minimization of KL divergence, score matching, as stipulated in Theorem 1 of \cite{10.5555/1795114.1795156}, aims to neutralize its derivative in the scale space at \( t = 0 \). This discussion underscores the impetus to extend score matching learning into a more flexible parametric learning framework, enhancing model stability and robustness.

As Corollary \ref{KL} shows, the gradient with respect to the noise covariance matrix \(\mathbf{\Sigma}\) is directly tied to score matching, as well.

\begin{corollary}\label{KL}
    Let \(P\), \(Q\), and \(\mathbf{N}\) be defined the same way as in Theorem \ref{f divergence}. Then:
\begin{align}\label{KL equation}
    \nabla_{\mathbf{\Sigma}}D_{\text{KL}}\bigl(p({\mathbf{y}}) \,\|\, q({\mathbf{y};\boldsymbol{\theta}})\bigr) 
    &= -\frac{1}{2}\,\mathcal{I}\bigl(p({\mathbf{y}}) \,\|\, q({\mathbf{y};\boldsymbol{\theta}})\bigr) ,
\end{align}
in which 
\[
D_{\text{KL}}\bigl(P \;\|\; Q\bigr) 
= \int \log \frac{dP}{dQ} dP
\]
is the classical Kullback-Leibler divergence for probability distribution $P \ll Q$.
\end{corollary}
\begin{proof}
    As such, Corollary \ref{KL} can be viewed as a subset derived from Theorem \ref{f divergence}, the proof is omitted here.
\end{proof}
When \(f(t) = t \log t\), the expression in \eqref{f divergence equation} simplifies to \eqref{KL equation} (as per Theorem \ref{KL}), thereby establishing Corollary \ref{KL} as a specific instance of Theorem \ref{f divergence}. Additionally, \eqref{KL equation} extends the high-dimensional de Bruijn's identity, which was originally detailed in \cite{1564431} with a focus on the local properties of differential entropy without addressing mismatched estimation. The findings from \cite{1564431} are encapsulated by:
\[
    \nabla_{\mathbf{\Sigma}}h(\mathbf{X} + \mathbf{N}) 
    = J(\mathbf{X} + \mathbf{N}),
\]
where \(h(\cdot)\) denotes the differential entropy. 

Thanks to Corollary~\ref{KL}, the gradient of KL divergence can be evaluated using score matching tools \cite{6795935,song2020sliced}, which are generally easy to compute for many input distributions (e.g., Gaussian mixture models and diffusion generative models). Hence gradient descent can be applied to find the optimal injected noise covariance matrix \[\mathbf{\Sigma} \leftarrow \mathbf{\Sigma} - \xi D_{\text{KL}}(p(\mathbf{y})\|q(\mathbf{y};\boldsymbol{\theta})),\] where $\xi$ is the step size. As evidenced by the applications of \eqref{KL equation} in numerous real-world problems \cite{huang2024blue,hoogeboom2023blurring,voleti2022score}, the methods developed here are also likely to be of practical relevance.

\section{Concluding Remarks}\label{implications}

Recent advancements in machine learning models have showcased the effectiveness of gradient-based optimization methods, particularly when paired with the Hessian matrix to employ Newton’s method \cite{1564431}, accelerating convergence and facilitating the rapid generation of correlated noise masks. This approach is especially effective in scenarios where the divergence function is concave with respect to system parameters \cite{5165186}, making it crucial for optimizing systems that generate samples by estimating covariance matrices \cite{yu2024constructing} or modeling Gaussian mixture models \cite{jia2023entropic}.

Additionally, our theorems support the development of blurring diffusion models \cite{hoogeboom2023blurring} and blue noise diffusion techniques \cite{huang2024blue}, aiming to understand and utilize correlated noise to reduce overfitting and enhance diversity in the generated outcomes, mirroring the advancements in correlated mask methods \cite{ulichney1993void}.


Moving beyond maximum likelihood training, we characterize density estimation rates of deep generative models via the metric entropy of hypothesis classes \cite{jia2023entropic}, providing a theoretical perspective on their generalization and training behavior. Implementing \(f\)-divergence as a training objective \cite{nowozin2016f}, substantially enhances generative model training, yielding images of higher fidelity. Score matching plays a crucial role here, allowing for a precise alignment of model predictions with the actual data distribution, thereby optimizing the training process. This enhanced method underscores the transformative impact of combining \(f\)-divergence with score matching, significantly advancing the capabilities of density estimation techniques in machine learning.


In this paper, we introduce \(f\) score matching, a versatile method for learning statistical diffusion models and estimating scores within implicit distributions. Our technique adeptly handles data with correlated noise and seamlessly integrates into contemporary generative model frameworks. Theoretically, our generalized relative Fisher information effectively bridges the score function and $f$-divergence under specific conditions, enhancing its applicability. This work extends the connection between the score function and maximum likelihood in the vector Gaussian channel, establishing new intersections between crucial information-theoretic and estimation-theoretic measures.

\section{Appendix}\label{appendix}

\begin{proof}[Proof of Lemma \ref{heat equation}]

For brevity, we omit references in the integrals whenever this does not cause ambiguity.
Recall that the relationship between the densities of channel input \( \mathbf{X} \) and output \( \mathbf{Y} \) as follows:
    \begin{align}
        &p(\mathbf{y}) = \int p(\mathbf{x}) (2\pi \tilde{\beta}_{t})^{-\frac{m}{2}} |\mathbf{\Sigma}|^{-\frac{1}{2}} \notag \\
        &\exp\left[-\frac{1}{2\tilde{\beta}_{t}}(\mathbf{y}-\sqrt{\tilde{\alpha}_t}\mathbf{Hx})^{T}\mathbf{\Sigma}^{-1}(\mathbf{y}-\sqrt{\tilde{\alpha}_t}\mathbf{Hx})\right] \, d\mathbf{x} \notag
    \end{align}
    where \( \tilde{\beta}_t = 1 - \tilde{\alpha}_t \) acts as the scale factor.
Let \(\mathbf{Q}(\mathbf{x},\mathbf{y})\) be a quadratic form associated with matrix \(\mathbf{\Sigma}\), and let \(\mathbf{G}(\mathbf{x},\mathbf{y})\) be the Gram matrix of 
\[
\mathbf{M} \;=\; \mathbf{\Sigma}^{-1} \bigl(\mathbf{y} - \sqrt{\tilde{\alpha}_t}\,\mathbf{H}\mathbf{x}\bigr).
\]
Concretely, we define
\begin{align}
   \mathbf{Q}(\mathbf{x},\mathbf{y}) 
   &= (\mathbf{y} - \sqrt{\tilde{\alpha}_t}\,\mathbf{H}\mathbf{x})^{T}\,\mathbf{\Sigma}^{-1}\,(\mathbf{y} - \sqrt{\tilde{\alpha}_t}\,\mathbf{H}\mathbf{x}), \notag
   \\
   \mathbf{G}(\mathbf{x},\mathbf{y})
   &= \mathbf{\Sigma}^{-1}\,(\mathbf{y} - \sqrt{\tilde{\alpha}_t}\,\mathbf{H}\mathbf{x})
      \,(\mathbf{y} - \sqrt{\tilde{\alpha}_t}\,\mathbf{H}\mathbf{x})^{T}\,\mathbf{\Sigma}^{-1}. \notag
\end{align}
Then, consider
\begin{equation}
    p(\mathbf{y}) = \int p(\mathbf{x})(2\pi \tilde{\beta}_{t})^{-\frac{m}{2}}\,
    |\mathbf{\Sigma}|^{-\frac12}\,
    \exp\! \ \Bigl(-\frac{1}{2\tilde{\beta}_{t}}\,\mathbf{Q}(\mathbf{x},\mathbf{y})\Bigr)
\,d\mathbf{x}. \notag
\end{equation}

\paragraph {Derivative w.r.t.\ \(\mathbf{\Sigma}\)}
We compute
\begin{align}
&\nabla_{\mathbf{\Sigma}}p(\mathbf{y}) \notag \\ 
&=
\int p(\mathbf{x})\,
\nabla_{\mathbf{\Sigma}}
\Bigl[
  (2\pi \tilde{\beta}_{t})^{-\frac{m}{2}}\,
  |\mathbf{\Sigma}|^{-\frac12}\,
  \exp\! \ \bigl(-\frac{1}{2\tilde{\beta}_{t}}\,\mathbf{Q}(\mathbf{x},\mathbf{y})\bigr)
\Bigr]
\,d\mathbf{x}. \notag
\end{align}

Standard matrix calculus yields
\begin{align}
    &\nabla_{\mathbf{\Sigma}}
\Bigl[
  |\mathbf{\Sigma}|^{-\frac12}\,\exp\! \ \bigl(-\frac{1}{2\tilde{\beta}_{t}}\,\mathbf{Q}\bigr)
\Bigr] \notag \\
&= 
|\mathbf{\Sigma}|^{-\frac12}\,\exp\! \ \bigl(-\frac{1}{2\tilde{\beta}_{t}}\,\mathbf{Q}\bigr)
\bigl(-\frac12\,\mathbf{\Sigma}^{-T}
  + \frac12\,\mathbf{G}(\mathbf{x},\mathbf{y})
\bigr), \notag
\end{align}
where \(\mathbf{\Sigma}^{-T} = (\mathbf{\Sigma}^{-1})^T\), which in the symmetric case simply becomes \(\mathbf{\Sigma}^{-1}\). Hence,
\begin{align}
\nabla_{\mathbf{\Sigma}} p(\mathbf{y})
&=\;
\int p(\mathbf{x}) \,
 (2\pi\tilde{\beta}_{t})^{-\frac{m}{2}}
 |\mathbf{\Sigma}|^{-\frac12}\, \notag \\
 &\exp\! \ \bigl(-\frac{1}{2\tilde{\beta}_{t}}\,\mathbf{Q}(\mathbf{x},\mathbf{y})\bigr)
 \bigl(
   -\frac12\,\mathbf{\Sigma}^{-T}
   + \frac12\,\mathbf{G}(\mathbf{x},\mathbf{y})
 \bigr)
\,d\mathbf{x} \notag
\\
&=\;
\int p(\mathbf{x}) \,
 (2\pi \tilde{\beta}_{t})^{-\frac{m}{2}}
 |\mathbf{\Sigma}|^{-\frac12}
 \Bigl[
   -\frac12\,\mathbf{\Sigma}^{-T}
   + \frac12\,\mathbf{G}(\mathbf{x},\mathbf{y})
 \Bigr] \notag \\
 &\exp\! \ \bigl(-\frac{1}{2\tilde{\beta}_{t}}\,\mathbf{Q}(\mathbf{x},\mathbf{y})\bigr)
\,d\mathbf{x}. \notag
\end{align}

\paragraph{First derivative w.r.t.\ \(\mathbf{y}\)}
Next,
\begin{align}
    &\nabla_{\mathbf{y}}p(\mathbf{y}) \notag \\ 
    &=
\int p(\mathbf{x})\,
\nabla_{\mathbf{y}}
\Bigl[
   (2\pi \tilde{\beta}_{t})^{-\frac{m}{2}}\,
   |\mathbf{\Sigma}|^{-\frac12}\,
   \exp\! \ \bigl(-\frac{1}{2\tilde{\beta}_{t}}\,\mathbf{Q}(\mathbf{x},\mathbf{y})\bigr)
\Bigr]
\,d\mathbf{x}. \notag
\end{align}

Since \((2\pi\tilde{\beta}_{t})^{-\tfrac m2}\,|\mathbf{\Sigma}|^{-\tfrac12}\) does not depend on \(\mathbf{y}\), only the exponential part matters. We know

\(\nabla_{\mathbf{y}}\exp\! \ \bigl(-\frac{1}{2\tilde{\beta}_{t}}\,\mathbf{Q}\bigr)
=
-\exp\! \ \bigl(-\frac{1}{2\tilde{\beta}_{t}}\,\mathbf{Q}\bigr)\,\mathbf{\Sigma}^{-1}(\mathbf{y} - \mathbf{H}\mathbf{x})\),
so we get
\begin{align}
    &\nabla_{\mathbf{y}} p(\mathbf{y}) \notag \\
&=\int p(\mathbf{x})\,
 (2\pi\tilde{\beta}_{t})^{-\frac{m}{2}}\,|\mathbf{\Sigma}|^{-\frac12}\,
 \bigl(-\,\mathbf{\Sigma}^{-1}(\mathbf{y}-\mathbf{H}\mathbf{x})\bigr)\, \notag \\
 &\exp\! \ \bigl(-\frac{1}{2\tilde{\beta}_{t}}\,\mathbf{Q}(\mathbf{x},\mathbf{y})\bigr)
\,d\mathbf{x}. \notag
\end{align}

\paragraph{Hessian (second derivative) w.r.t.\ \(\mathbf{y}\)}
The Hessian \(\mathbb{H}_\mathbf{y} p(\mathbf{y})\) is
\begin{align}
    &\mathbb{H}_\mathbf{y} p(\mathbf{y})
\;=\;
\nabla_{\mathbf{y}}\!\bigl[
   \nabla_{\mathbf{y}}^{T}\,p(\mathbf{y})
\bigr] \notag \\
&= \int p(\mathbf{x})\,
 (2\pi\tilde{\beta}_{t})^{-\frac{m}{2}}\,|\mathbf{\Sigma}|^{-\frac12}\, \notag \\
 &\nabla_{\mathbf{y}}^{T}\!
 \Bigl[
   -\,\mathbf{\Sigma}^{-1}(\mathbf{y}-\mathbf{H}\mathbf{x})
   \,\exp\! \ \bigl(-\frac{1}{2\tilde{\beta}_{t}}\,\mathbf{Q}(\mathbf{x},\mathbf{y})\bigr)
 \Bigr]
\,d\mathbf{x}. \notag
\end{align}
Applying product rule and simplifying (similarly to the \(\mathbf{\Sigma}\) derivative case) gives us
\begin{align}
    &\mathbb{H}_\mathbf{y} p(\mathbf{y}) =\int p(\mathbf{x})\,
 (2\pi\tilde{\beta}_{t})^{-\frac{m}{2}}\,|\mathbf{\Sigma}|^{-\frac12}\, \notag \\
 &\bigl[\,-\,\mathbf{\Sigma}^{-T} \;+\; \mathbf{G}(\mathbf{x},\mathbf{y})\bigr]\, 
 \exp\! \ \bigl(-\frac{1}{2\tilde{\beta}_{t}}\,\mathbf{Q}(\mathbf{x},\mathbf{y})\bigr)
\,d\mathbf{x}. \notag
\end{align}

\paragraph{Relation between \(\nabla_{\mathbf{\Sigma}} p(\mathbf{y})\) and \(\mathbb{H}_\mathbf{y} p(\mathbf{y})\)}
Comparing this with the expression of \(\nabla_{\mathbf{\Sigma}} p(\mathbf{y})\), we observe that
\begin{align}
    \nabla_{\mathbf{\Sigma}} p(\mathbf{y})
&= \int p(\mathbf{x})\,
 (2\pi\tilde{\beta}_{t})^{-\frac{m}{2}}\,|\mathbf{\Sigma}|^{-\frac12} \notag \\ 
 &\Bigl[\,
   -\tfrac12\,\mathbf{\Sigma}^{-T}
   \;+\;\frac12\,\mathbf{G}(\mathbf{x},\mathbf{y})
 \Bigr]\, 
 \exp\! \ \bigl(-\frac{1}{2\tilde{\beta}_{t}}\,\mathbf{Q}(\mathbf{x},\mathbf{y})\bigr)
\,d\mathbf{x}, \notag
\end{align}
while
\begin{align}
    &\mathbb{H}_\mathbf{y} p(\mathbf{y})
\;=\;
\int p(\mathbf{x})\,
 (2\pi\tilde{\beta}_{t})^{-\frac{m}{2}}\,|\mathbf{\Sigma}|^{-\frac12} \notag \\
 &\Bigl[\,
   -\mathbf{\Sigma}^{-T}
   \;+\;
   \mathbf{G}(\mathbf{x},\mathbf{y})
 \Bigr]\,
 \exp\! \ \bigl(-\frac{1}{2\tilde{\beta}_{t}}\,\mathbf{Q}(\mathbf{x},\mathbf{y})\bigr)
\,d\mathbf{x}. \notag
\end{align}
Hence it follows that
\begin{equation}
    \nabla_{\mathbf{\Sigma}} p(\mathbf{y})
\;=\;
\tfrac12\,\mathbb{H}_\mathbf{y}\,p(\mathbf{y})
\;\equiv\;
\tfrac12\,\Delta_{\mathbf{y}}\,p(\mathbf{y}).
\end{equation}
\end{proof}

\begin{proof}[Proof of Lemma \ref{integration vanish}]
    Let \(V_r\subset \mathbb{R}^m\) be the region (an \(m\)-dimensional ball) bounded by the closed, piecewise-smooth, oriented surface \(S_r\), which is the \emph{\(m\)-sphere} of radius \(r\) centered at the origin. At any point \(\mathbf{y}\in S_r\), the symbol \(\mathbf{e}_{S_r}(\mathbf{y})\) denotes the outward-pointing unit normal vector to \(S_r\). Under the notation $d\mathbf{s}_{r} = \|d\mathbf{s}_{r}\|\mathbf{e}_{S_r}(\mathbf{y})$.

We consider the integral:
\begin{equation}\label{50}
    \int_{\mathbb{R}^m} \Delta_{\mathbf{y}} \left( q(\mathbf{y};{\boldsymbol{\theta}}) f\left(\frac{p(\mathbf{y})}{q(\mathbf{y};{\boldsymbol{\theta}})}\right) \right) d\mathbf{y}. \notag
\end{equation}
Applying Gauss's theorem, we rewrite this integral as a surface integral:
\begin{align}\label{surface integration}
    &\int_{\mathbb{R}^m} \, \nabla_{\mathbf{y}}^{2}(qf(\frac{p}{q})) d\mathbf{y} = \lim_{r\rightarrow\infty} \int_{V_r} \nabla_{\mathbf{y}} \cdot \nabla_{\mathbf{y}}(qf(\frac{p}{q})) d\mathbf{y} \notag \\ 
    &= \lim_{r\rightarrow\infty}\int_{S_{r}} \left(\nabla_{\mathbf{y}} qf(\frac{p}{q})\cdot \mathbf{e}_{S_{r}} (\mathbf{y})  \right)dS_{r}. \notag
\end{align}
Collect terms, noting \(f\)-divergence is finite, to obtain:
\begin{align}
     \lim_{r \rightarrow \infty} \int_{S_r}\left(\nabla_{\mathbf{y}} qf(\frac{p}{q})  \right) \cdot d\mathbf{s}_{r} = 0. \notag
\end{align}
We integrate over $r \geq 0$ the surface integral above and apply Green’s identity to find the relations
\begin{align}
    &\int_{0}^\infty \int_{S_{r}} \left((f(\frac{p}{q}))\nabla_{\mathbf{y}}q  + q \nabla_{\mathbf{y}}f(\frac{p}{q})  \right)\cdot d\mathbf{s}_{r}  d r \notag \\ &= \lim_{r\rightarrow\infty} \int_{S_r} qf(\frac{p}{q})\mathbf{e}_{S_{r}}(\mathbf{y}) \cdot d\mathbf{s}_{r}
    - \int_{\mathbb{R}^{m}} \nabla_{\mathbf{y}} \cdot (qf(\frac{p}{q})) d\mathbf{y}. \notag
\end{align}
The surface integral above vanishes in the limit \(r \to \infty\) since the finite $f$-divergence, the corresponding volume integral over all of \(\mathbb{R}^m\) must be finite due to \(p(\mathbf{y})\) and \(q(\mathbf{y};{\boldsymbol{\theta}})\) are smooth and decay quickly:
\[
\left|\, 
\int_{\mathbb{R}^{n}} \nabla_{\mathbf{y}} \cdot \left( q\, f\left( \frac{p}{q} \right) \right) \, d\mathbf{y} 
\,\right| \leq \infty .
\]
Hence, on \(\mathbb{R}^m\), together with the fact that the limit exists, this proves that the entire term in \eqref{integration vanish equation} evaluates to zero:
\[
\frac{1}{2}\int_{\mathbb{R}^m} 
\nabla_{\mathbf{y}}\,\cdot\,
\bigl[\nabla_{\mathbf{y}}
\bigl(q\,f(\frac{p}{q})\bigr)\bigr] \,d\mathbf{y}
\;=\; 0,
\]
and therefore the first term of \eqref{last term} integrates to zero.
\end{proof}

\newpage
\bibliographystyle{IEEEtran}

\bibliography{reference}

\end{document}